\documentclass[sigconf,nonacm]{aamas} 

\usepackage{balance} 
\usepackage{tikz}
\usepackage{graphicx}
\usepackage{float}
\usepackage{dblfloatfix}
\usepackage{placeins}
\usepackage{afterpage}
\usepackage{algorithm}
\usepackage[noend]{algpseudocode}
\usepackage{amsmath}
\usepackage{amsthm}

\newtheorem{theorem}{Theorem}[section]
\newtheorem{observation}[theorem]{Observation}

\setcopyright{ifaamas}
\acmConference[AAMAS '25]{Proc.\@ of the 24th International Conference
on Autonomous Agents and Multiagent Systems (AAMAS 2025)}{May 19 -- 23, 2025}
{Detroit, Michigan, USA}{A.~El~Fallah~Seghrouchni, Y.~Vorobeychik, S.~Das, A.~Nowe (eds.)}
\copyrightyear{2025}
\acmYear{2025}
\acmDOI{}
\acmPrice{}
\acmISBN{}

\acmSubmissionID{<<EasyChair submission id>>}

\title[AAMAS-2025 Formatting Instructions]{Interval Selection with Binary Predictions}

\author{Christodoulos Karavasilis}
\affiliation{
  \institution{University of Toronto}
  \city{Toronto}
  \country{Canada}}
\email{ckar@cs.toronto.edu}

\begin{abstract}
Following a line of work that takes advantage of vast machine-learned data to enhance online algorithms with (possibly erroneous) information about future inputs, we consider predictions in the context of deterministic algorithms for the problem of selecting a maximum weight independent set of intervals arriving on the real line. We look at two weight functions, unit (constant) weights, and weights proportional to the interval's length. In the classical online model of irrevocable decisions, no algorithm can achieve constant competitiveness (Bachmann et al. \cite{bachmann2013online} for unit, Lipton and Tomkins \cite{lipton1994online} for proportional). In this setting, we show that a simple algorithm that is faithful to the predictions is optimal, and achieves an objective value of at least $OPT-\eta$, with $\eta$ being the total error in the predictions, both for unit, and proportional weights.

When revocable acceptances (a form of \textit{preemption}) are allowed, the optimal deterministic algorithm for unit weights is $2k$-competitive \cite{borodin2023any}, where $k$ is the number of different interval lengths. We give an algorithm with performance $OPT -\eta$ (and therefore $1$-consistent), that is also $(2k+1)$-robust. For proportional weights, Garay et al. \cite{garay1997efficient} give an optimal $(2\phi +1)$-competitive algorithm, where $\phi$ is the golden ratio. We present an algorithm with parameter $\lambda > 1$ that is $\frac{3\lambda}{\lambda -1}$-consistent, and $\frac{4\lambda^2 + 2\lambda}{\lambda -1}$-robust. Although these bounds are not tight, we show that for $\lambda > 3.42$ we achieve consistency better than the optimal online guarantee in \cite{garay1997efficient}, while maintaining bounded robustness.

 We conclude with some experimental results on real-world data that complement our theoretical findings, and show the benefit of prediction algorithms for online interval selection, even in the presence of high error.
\end{abstract}


\newcommand{\BibTeX}{\rm B\kern-.05em{\sc i\kern-.025em b}\kern-.08em\TeX}

\begin{document}


\pagestyle{fancy}
\fancyhead{}


\maketitle 


\section{Introduction}
We consider the problem of online \textit{interval selection}, or \textit{interval scheduling} on a single machine, where real-length intervals arrive online, and we must output a set of non-conflicting intervals. Each interval is associated with a weight, and the goal is to maximize the sum of weights of the intervals in the solution. This problem is equivalent to finding a maximum weight independent set in interval graphs. We focus on two weight functions, \textit{unit} (or constant) weights, and \textit{proportional} weights, where the weight of an interval is equal to its length.  While interval scheduling is often studied under the real-time assumption where intervals arrive in order of non-decreasing starting times, we consider the generalized version of \textit{any-order} arrivals \cite{borodin2023any}. In the traditional online model of irrevocable decisions, no algorithm (even randomized) can achieve a constant competitive ratio (Bachmann et al. \cite{bachmann2013online} for unit weights, Lipton and Tomkins \cite{lipton1994online} for proportional). Because of this, and because some applications permit it, a relaxation of the problem that allows for revocable acceptances has been considered. In this model, every new interval can be accepted by displacing any conflicting intervals in the solution, but every rejection is final. In the area of scheduling, this is sometimes also called \textit{preemption}, although no ``restarts'' are allowed in our problem. In the offline setting, an optimal solution can be easily found in polynomial time, both for unit, and for proportional weights \cite{kleinberg2006algorithm}. The applications of interval scheduling include routing \cite{plotkin1995competitive}, computer wiring \cite{gupta1979optimal}, project selections during space missions \cite{hall1994maximizing}, and satellite photography \cite{gabrel1995scheduling}. A more detailed discussion on the applications of interval scheduling can be found in the surveys by Kolen et al. \cite{kolen2007interval} and Kovalyov et al. \cite{kovalyov2007fixed}.\\\\
Motivated by advancements in machine learning and access to a plethora of data, there has been an effort to equip online algorithms with possibly erroneous predictions about the input instance. Such algorithms are able to achieve much better performance when these predictions are accurate, overcoming some pessimistic bounds of competitive analysis, and helping to bridge the gap between theory and practice. Various classical online problems such as ski rental and non-clairvoyant job scheduling \cite{purohit2018improving}, caching \cite{lykouris2021competitive}, facility location \cite{almanza2021online}, metrical task systems \cite{antoniadis2023online}, and matching \cite{antoniadis2020secretary} have been considered in this model. See Mitzenmacher and Vassilvitskii \cite{DBLP:books/cu/20/MitzenmacherV20} for a more detailed survey on the topic, and \cite{ALPS} for an online repository of relevant papers. Predictions are also a form of \textit{untrusted advice} (Angelopoulos et al. \cite{angelopoulos2024online}), a natural extension of the model of online algorithms with advice (Boyar et al. \cite{boyar2017online}) when the advice is imperfect. Advice research tends to be more information theoretic, focusing on tradeoffs between the number of advice bits and the quality of the solution. Although predictions are often available as offline information, given to the algorithm in advance, we consider a model where a prediction is associated with each input item, and is also given online. This is quite natural and has been considered before for problems such as paging, graph coloring, and packing (\cite{lykouris2021competitive,rohatgi2020near,antoniadis2023paging,antoniadis2024online,grigorescu2024simple}). This setting also allows for an oracle to adapt as more of the input is revealed, enabling research where there are different bounds on the quality of later predictions, and allowing one to tailor the predictor algorithm directly \cite{elias2024learning}. Furthermore, we use binary predictions, which has our model falling in line with work considering limited size, or \textit{succinct} predictions \cite{antoniadis2023paging,berg2024complexity,angelopoulos2023contract}.\\\\
\textbf{Related work.} Table \ref{tab:prev_work} shows the most relevant existing work in the conventional online setting. In the case of irrevocable decisions, no algorithm (even randomized) can achieve a constant competitive ratio. For the relaxed model of revocable acceptances, we use an asterisk to indicate that the competitive ratio is optimal. In the context of randomized algorithms and revocable acceptances, Emek at al. \cite{emek2016space} give a $6$-competitive algorithm for unit weights, while we know of no work improving upon the $(2\phi + 1)$-competitive algorithm.

\begin{table}[h]\centering
	\caption{Online results without predictions: $n$ is the size of the input, $k$ is the number of different lengths, $\Delta$ is the ratio of the longest to shortest interval.}
	\label{tab:prev_work}
	\begin{tabular}{c  c  c}
		 & \textit{Unit} & \textit{Proportional} \\ \toprule
		\parbox[c]{2cm}{Irrevocable \\ (randomized)}  & $\Omega(n)$ \cite{bachmann2013online} & $\Omega(\log\Delta)$ \cite{lipton1994online} \\ \midrule
		\parbox[c]{2cm}{Revocable \\ (deterministic)} & $2k^*$ \cite{borodin2023any} & $(2\phi + 1)^*$ \cite{garay1997efficient,tomkins1995lower}  \\
  \bottomrule
	\end{tabular}
\end{table}

Boyar et al. \cite{boyar2023online} is the most closely related work to our problem with predictions, and motivated our study. They consider the case of unit weighted intervals on a line graph, and give an optimal deterministic algorithm in the setting of irrevocable decisions with performance $OPT - \eta$ for a different set of predictions and error measure. We extend their work using (possibly adaptive) predictions of limited size, considering an additional weight function of interest, and initiating the study of these problems with revocable decisions.\\\\
\textit{Structure of the paper.} In section \ref{section:prelim} we formally define the model, including our predictions and error measure. Section \ref{section:irrev} is about the model of irrevocable decisions, whereas in section \ref{section:rev} we allow for revocable acceptances. We conclude with some experiments on real-world data (section \ref{section:exp}) that showcase the usefulness of our predictions, and complement our theoretical results.


\section{Problem Setting, Definitions and Discussion} \label{section:prelim}
In the problem of \textit{interval selection}, an instance consists of a set of intervals arriving on the real line. Each interval is specified by a starting point $s_i$ and an end point $f_i$, with $s_i < f_i$, and it occupies space $[s_i,f_i)$. The conventional notions of  \textit{intersection}, \textit{disjointness}, and \textit{containment} apply. Two intervals can conflict because of a \textit{partial conflict}, or because of \textit{proper inclusion} (figure \ref{fig:conflicts}). In the latter case, we say that the smaller (larger) interval is subsumed (subsumes) by the other. Each interval $I$ is also associated with a weight $w(I)$. The goal is to output a set of disjoint intervals of maximum weight. We focus on two types of weight functions, \textit{unit} weights where $w(I) = 1$ for all $I$, and \textit{proportional} weights where $w(I) = f_i - s_i$. With unit weights, we want to accept as many intervals as possible, whereas with proportional weights, we want to cover as much of the line as possible.
\begin{figure}[H]
	\centering
	
	\begin{tikzpicture}[scale=0.75]

	\node[draw=none] (I1a) at (-10,0) {$ $};
	\node[draw=none] (I1b) at (-6,0) {$ $};
	\draw[line width=0.5mm] (I1a) -- (I1b);


	\node[draw=none] (I2a) at (-7,0.5) {$ $};
	\node[draw=none] (I2b) at (-3,0.5) {$ $};
	\draw[line width=0.5mm] (I2a) -- (I2b);
	
	\node at (-5,-1) {$(a) \text{ Partial Conflict.}$};
	
	\node[draw=none] (I3a) at (-6,-3.5) {$ $};
	\node[draw=none] (I3b) at (-4,-3.5) {$ $};
	\draw[line width=0.5mm] (I3a) -- (I3b);


	\node[draw=none] (I4a) at (-9,-3) {$ $};
	\node[draw=none] (I4b) at (-1,-3) {$ $};
	\draw[line width=0.5mm] (I4a) -- (I4b);
	
	\node at (-5,-4.5) {$(b) \text{ Proper inclusion conflict}.$};

	\end{tikzpicture} 
	\caption{Types of conflicts.}\label{fig:conflicts}
\end{figure}
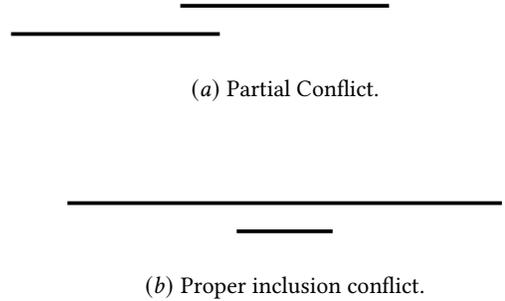

The sequence of the intervals arriving online is chosen by an adversary with full knowledge of the algorithm. In the conventional model of \textit{irrevocable decisions}, each accepted interval is final and will be part of the final solution. A well studied relaxation of the model is allowing for \textit{revocable acceptances}, where any new interval may be accepted by displacing any conflicting intervals in the solution, but every rejection is final. This is sometimes also referred to as \textit{preemption}.\\\\
We measure the performance of an online algorithm using \textit{worst case competitive analysis} \cite{BorOnlineBook}. Let $ALG$ denote the total weight of the algorithm's solution, and $OPT$ be the total weight of an optimal solution. An algorithm is strictly $c$-competitive, if for all instances and input permutations, we have that $\frac{OPT}{ALG}\geq c$. We also refer to $ALG$ as the \textit{performance} of the algorithm. Lastly, we will use $ALG$ (respectively $OPT$) to denote the set of intervals in the algorithm's (optimal) solution. The meaning should always be clear from context.\\\\
\textbf{Predictions.} Every interval is associated with a binary prediction that becomes known at the time of the interval's arrival, denoting whether or not that interval is part of some fixed optimal solution. More precisely, $Prd(I)=1$ means interval $I$ is predicted to be optimal, and $Prd(I)=0$ means that it is predicted to not be optimal. Let $\mathcal{I}$ denote the set of all intervals in the instance. We define $\overrightarrow{\textbf{p}} = (Prd(I_1),Prd(I_2),...,Prd(I_{|\mathcal{I}|}))$ to be the binary vector of all the online predictions. Let also $\overrightarrow{\textbf{p}_{+}}$ be the predictions vector when all predictions are accurate, and $\overrightarrow{\textbf{p}_{-}}= \overline{\overrightarrow{\textbf{p}_{+}}}$.\\
\textbf{Error.} Every inaccurate prediction introduces an amount of error. Let $\eta(I)$ be the amount of error introduced by interval $I$. If the prediction of $I$ was accurate, we define $\eta(I)=0$. If $I$ was wrongly predicted to be non-optimal, let $\eta(I) = w(I)$, and if $I$ was wrongly predicted to be optimal, and $\mathcal{C}$ is the set of optimal intervals $I$ conflicts with, let $\eta(I) =\sum_{J\in \mathcal{C}}w(J)\;-\;w(I)$. Let the total error be $\eta = \sum_{I\in\mathcal{I}}\eta(I)$. One may fix any optimal solution to measure the error against. We use $\eta_{max}$ to denote the maximum possible error, i.e. $\eta_{max}=\eta$ when $\overrightarrow{\textbf{p}} = \overrightarrow{\textbf{p}_{-}}$.\\\\
A common approach to evaluate algorithms that use predictions is to focus on an algorithm's \textit{consistency} and \textit{robustness}. We say that an algorithm is $\gamma$-consistent if it is $\gamma$-competitive, whenever the predictions are accurate, i.e. $\overrightarrow{\textbf{p}} = \overrightarrow{\textbf{p}_{+}}$. We say that an algorithm is $\zeta$-robust, if it is $\zeta$-competitive regardless of the accuracy of the predictions. There is usually a trade-off between consistency and robustness, and the goal is to design algorithms with consistency close to $1$, and robustness that is not far worse than the competitiveness of the best predictionless, online algorithm.\\\\
Most of our proofs work by using a \textit{charging argument}, where we map the weight of optimal intervals to the weight of intervals taken by the algorithm, and sometimes error. This charging is usually defined in an online manner, and throughout the execution of the algorithm, we use $\Phi(I)$ to refer to the total amount of charge to interval $I$. In the model of revocable acceptances, we will distinguish between \textit{direct}, and \textit{transfer} charging. Transfer charging ($TC$) occurs at the moment a new interval is accepted by replacing existing intervals, and refers to the amount of charge it inherits because of this. Direct charge ($DC$) takes place afterwards, whenever an interval causes optimal intervals to be rejected. We use $TC(I)$ (respectively $DC(I)$) to denote the amount of transfer (direct) charge of an interval, with $\Phi(I) = TC(I) + DC(I)$.\\
We also use the notion of a \textit{predecessor trace}, which is analogous to Woeginger's \cite{woeginger1994line} \textit{predecessor chain} in th real-time model. If $I$ is an interval in the algorithm's solution, the \textbf{predecessor trace} $\mathcal{P}$ of $I$ is the maximal list of intervals $(P_1,P_2,..,P_k = I)$, such that $P_i$ was at some point accepted by the algorithm, but was later replaced by $P_{i+1}$. 


\section{Irrevocable Acceptances}\label{section:irrev}
In utilizing access to predictions, a natural algorithm to first consider is one that simply \textit{follows the predictions}. We therefore present algorithm \ref{alg:naive}, which is the main subject of this section.

\begin{algorithm}
\caption{\texttt{Naive}}\label{alg:naive}
\begin{algorithmic}
\State On the arrival of $I$:
\State $I_{s} \gets $ Set of intervals currently in the solution conflicting with $I$
\If{$Prd(I) = 1$ and $I_{s} = \emptyset$}
    \State Take $I$
\EndIf
\end{algorithmic}
\end{algorithm}
\textbf{Unit weights.} We first consider the case of unit weights, and prove the following (positive) result on the performance of algorithm \texttt{Naive}.

\begin{theorem}
    Algorithm Naive achieves $ALG \geq OPT - \eta$ for interval selection with unit weights.
    \label{theo:unw-naive-pos}
\end{theorem}
\begin{proof}
The proof works by mapping optimal intervals to error, and to intervals taken by the algorithm, given that every missed optimal interval can be associated with at least one unit of error. Let $OPT$ be an optimal solution, and let $ALG$ be the algorithm's solution. For each unit of error, we define an error element $h$.
Let $H=\{h_{1},...,h_{\eta}\}$ be the set of error elements. Let $H_{I} \subseteq H$, be the set of error elements corresponding to the error $\eta(I)$. It holds that $H_{I} \cap H_{J} = \emptyset$ for any two distinct intervals $I,J$, and $\bigcup\limits_{I} H_{I} = H$.\\\\
We define an injective mapping $F: OPT \rightarrow ALG \cup H$ as follows: Let $I_{opt}$ be an optimal interval. If $I_{opt}$ is taken by the algorithm, it is mapped onto itself. If $I_{opt}$ is not taken by the algorithm, there are two possibilities. The first possibility is that $Prd(I_{opt}) = 1$, but $I_{opt}$ conflicted with another interval $I_{c}$ taken by the algorithm. For $I_{c}$ to have been taken, it means that $Prd(I_{c}) = 1$, and $|H_{I_{c}} \cup \{I_c\}|>0$ because $I_c$ conflicts with $I_{opt}$. In this case we map $I_{opt}$ to an error element $h_{c}\in H_{I_{c}}$, or $I_c$ itself. Even if more optimal intervals were not taken because they conflicted with $I_{c}$, there will be enough distinct elements in $H_{I_{c}} \cup \{I_c\}$ to map them to.\\\\
The second possibility is that $I_{opt}$ was not taken, because $Prd(I_{opt}) = 0$. In this case, we have that $|H_{I_{opt}}| = 1$, and we can map $I_{opt}$ to the error element of its own prediction. In conclusion, we have that $ALG + |H| \geq OPT$, and we get the desired bound.    
\end{proof}
\begin{corollary}
Algorithm Naive is $1$-consistent.
\end{corollary}

\begin{theorem}
    For every deterministic algorithm, there exist a unit weights instance and predictions, such that $ALG = OPT - \eta$.
    \label{theo:unw-naive-neg}
\end{theorem}
\begin{proof}
    Let an interval $I_{big}$ arrive first, with $Prd(I_{big}) = 0$. If the algorithm rejects it, no more intervals arrive, and we have $ALG = 0$, $OPT = 1$, $\eta = 1$. If the algorithm takes $I_{big}$, two non-conflicting intervals arrive next, $I_{1}$ and $I_{2}$, both subsumed by $I_{big}$, with $Prd(I_{1}) = 0$ and $Prd(I_{2})=1$. In this case, we have $ALG = 1$, $OPT = 2$, $\eta = \eta(I_1) = 1$. In both cases, the equality holds. One can repeat this construction for an asymptotic result.
    \begin{figure}[h] 
        \centering
        \begin{tikzpicture}[scale=0.5]
	
	\node at (0,0.5) {$Prd(I_{big}) = 0$};
	\node[draw=none] (I1a) at (-8,0) {$ $};
	\node[draw=none] (I1b) at (8,0) {$ $};
	\draw[line width=0.5mm] (I1a) -- (I1b);

	\node[draw=none] (I5a) at (-6,-1) {$ $};
	\node[draw=none] (I5b) at (-1,-1) {$ $};
	\node at (-3.5,-2) {$Prd(I_{1})=0$};
	\draw[color=red,line width=0.5mm] (I5a) -- (I5b);

 \node[draw=none] (I6a) at (0,-1) {$ $};
	\node[draw=none] (I6b) at (5,-1) {$ $};
	\node at (2.5,-2) {$Prd(I_{2})=1$};
	\draw[color=red,line width=0.5mm] (I6a) -- (I6b);

	\end{tikzpicture} 
        \caption{Instance of theorem \ref{theo:unw-naive-neg}.}
        \label{fig:neg-unit-irrev}
    \end{figure}
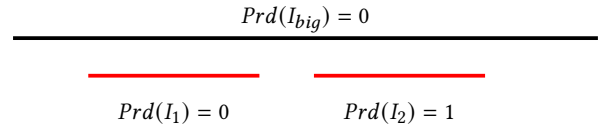
\end{proof}

\begin{corollary}
    Algorithm Naive is optimal for unit weights in the model of irrevocable acceptances.
\end{corollary}
Boyar et al. \cite{boyar2023online} were the first to consider the problem of interval selection with unit weights and irrevocable decisions, and they get the same (syntactically) performance, using a different algorithm, and a different set of predictions and error measure. In comparing our result to theirs, we note that our predictions are information theoretically strictly weaker than theirs\footnote{Their predictions consist of the entire input instance given in advance.}, and can in fact easily be extracted from theirs, allowing our algorithms to operate in their model. Furthermore, their predictions-following algorithm is enhanced with a \textit{greedy} aspect in order to achieve this optimal performance. As we will see in section \ref{section:exp}, experimental results on real-world data suggest that for some error ranges, pure \textit{greediness} is arguably a more important attribute than the use of predictions for getting a good solution, and the combination of both in the context of revocable acceptances works best.\\\\

We will now show that with \textbf{proportional weights}, algorithm \texttt{Naive} achieves the same performance bounds as in the case of unit weights.

\begin{theorem}
Algorithm Naive achieves $ALG \geq OPT - \eta$ for interval selection with proportional weights.
    \label{theo:prop-naive-pos}
\end{theorem}
\begin{proof}
        Without loss of generality, we assume integral lengths of intervals, and later explain how to generalize to real lengths. We discretize the weight of intervals into weight units, and define a weight element $w$ for each unit of weight. Let $W_{I} = \{w_{1},...,w_{w(I)}\}$ be the set of weight elements corresponding to the weight of interval $I$, with $W_{I} \cap W_{J} = \emptyset$ for any two distinct intervals $I,J$ . Let $W_{opt} = \bigcup_{I\in OPT}W_{I}$, and $W_{alg} = \bigcup_{I\in ALG}W_{I}$. The sets $H$ and $H_{I}$ are defined as for the unweighted algorithms. Lastly, let $C_{opt}(I)$ be the set of optimal intervals in $OPT$ that conflict with interval $I$.\\\\
    We argue for the existence of an injective mapping $F: W_{opt} \rightarrow W_{alg} \cup H$ as follows: Let $I_{opt}$ be an optimal interval. If $I_{opt}$ is taken by the algorithm, we map the elements of $W_{I_{opt}}$ to their corresponding elements in $W_{alg}$. If $I_{opt}$ is not taken by the algorithm, there are two cases. The first case is that $I_{opt}$ did not conflict with any interval in the solution, but $Prd(I_{opt}) = 0$. In this case, we know that $\eta(I_{opt}) = |H_{I_{opt}}| = w(I_{opt})$, and we can map the weight elements of $I_{opt}$ to error elements in $H_{I_{opt}}$.\\\\
    The other case is that $I_{opt}$ conflicted with at least one interval in the solution. Let that conflicting interval be $I_{c}$. It could also be that $Prd(I_{opt}) = 0$, and $H_{I_{opt}}$ would have error elements we can use, but we will assume the worst case of $Prd(I_{opt}) = 1$. In this case, it could be  $|W_{I_{opt}}| >  |H_{I_{c}}|$, and we cannot map the weight elements to error elements exclusively. We can, however, map the optimal weight elements, to elements in $H_{I_{c}} \cup W_{I_{c}}$. To see that there will always be sufficiently many unmapped elements, notice that $|H_{I_{c}} \cup W_{I_{c}}| \geq |\bigcup_{I \in C_{opt}(I_{c})}W_{I}|$. This is because $|H_{I_{c}}| = |\bigcup_{I \in C_{opt}(I_{c})}W_{I}| - |W_I|$, and $W_{I_c} \cap H_{I_c} = \emptyset$ always holds. We conclude that $|W_{alg}| + |H| \geq |W_{opt}|$, and we get the desired bound.\\\\
    To adapt the proof to real lengths, instead of considering sets of error and weight elements, we can define a transport plan using two transport matrices $H$ and $W$ of size $OPT\times ALG$. $H_{ij}$ (respectively $W_{ij}$) corresponds to the (real) amount of weight, or mass, mapped from $I_i \in OPT$ to the amount of error (resp. weight) introduced by $I_j \in ALG$. We can define these matrices such that for $1 \leq i \leq OPT$, $\sum_{1\leq j \leq ALG} H_{ij} + W_{ij} = w(I_i)$, for $1\leq j \leq ALG$, we have that $\sum_{1\leq i \leq OPT}H_{ij} \leq \eta(I_j)$ and $\sum_{1\leq i \leq OPT}W_{ij} \leq w(I_j)$.
\end{proof}
\begin{theorem}
For every deterministic algorithm, there exist a proportional weights instance and predictions, such that $ALG = OPT - \eta$.
    \label{theo:prop-naive-neg}
\end{theorem}
\begin{proof}
    Let $I_{1}$ arrive first with $Prd(I_{1}) = 0$. If the algorithm doesn't accept $I_{1}$, no more intervals arrive, and we have that $ALG = 0$, $OPT = w(I_{1})$, and $\eta = w(I_{1})$. If the algorithm accepts $I_{1}$, let two intervals $I_2$ and $I_3$ arrive next, with $w(I_2) = w(I_1)$, $Prd(I_2) = 1$, $w(I_3) = 2w(I_1)$ and $Prd(I_3) = 0$. In this case we have $ALG = w(I_1)$, $OPT = 3w(I_1)$, and $\eta = \eta(I_3) = 2w(I_1)$. In both cases the equality holds. One can repeat this construction for an asymptotic result.
\end{proof}

    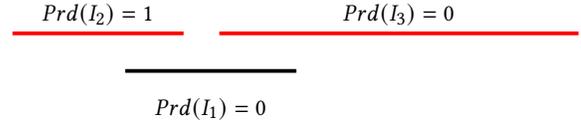
\begin{figure}
	\centering
	
	\begin{tikzpicture}[scale=0.5]

	\node at (-5,-2) {$Prd(I_{1}) = 0$};
	\node[draw=none] (I1a) at (-7.5,-1) {$ $};
	\node[draw=none] (I1b) at (-2.5,-1) {$ $};
	\draw[line width=0.5mm] (I1a) -- (I1b);


	\node[draw=none] (I3a) at (-10.5,0) {$ $};
	\node[draw=none] (I3b) at (-5.5,0) {$ $};
	\node at (-8,0.5) {$Prd(I_{2}) = 1$};
	\draw[color=red,line width=0.5mm] (I3a) -- (I3b);
	
	\node[draw=none] (I4a) at (-5,0) {$ $};
	\node[draw=none] (I4b) at (5,0) {$ $};
	\node at (0,0.5) {$Prd(I_{3}) = 0$};
	\draw[color=red,line width=0.5mm] (I4a) -- (I4b);

	\end{tikzpicture} 
	\caption{Instance of Theorem \ref{theo:prop-naive-neg}, with $w(I_2) = w(I_1), w(I_3) = 2w(I_1)$.}\label{fig:prop-neg-no-rev}
\end{figure}

\begin{corollary}
 Algorithm Naive is optimal for proportional weights in the model of irrevocable acceptances.
\end{corollary}
\section{Revocable Acceptances}\label{section:rev}
Given the difficulty of the problem(s) in the conventional online model, we now consider the case where acceptances are revocable, but rejections are final, a relaxation of the model that is commonly studied for the problem of interval selection. A new interval can now always be accepted by displacing any intervals in the solution conflicting with it. For unit weights, Borodin and Karavasilis \cite{borodin2023any} give an optimal algorithm that is $2k$-competitive, where $k$ is the number of distinct interval lengths. We will refer to this algorithm of \cite{borodin2023any} as the \textit{BK2K} algorithm. \textit{BK2K} is a greedy algorithm, always accepting a new interval when there is no conflict, and whenever a conflict exists, the new interval is accepted only if it is properly included in an interval currently in the solution. We use that as the base logic for our predictions algorithm, and add one more replacement rule, which accepts a new interval $I$ that is only involved in partial conflicts, if $Prd(I) = 1$. Furthermore, an interval accepted by that rule gets marked, to make sure it cannot be replaced by that rule again. We call this algorithm \texttt{Revoke-Unit} \ref{alg:unw-revoke}. Interestingly, this rule of locally \textit{following the predictions once}, suffices to give us $1$-consistency.

\begin{algorithm}
\caption{\texttt{Revoke-Unit}}\label{alg:unw-revoke}
\begin{algorithmic}
\State $M \gets \emptyset$ \Comment{Set of marked intervals}
\State $S \gets \emptyset$ \Comment{Solution set}
\State On the arrival of $I$:
\State $I_{s} \gets $ Set of intervals currently in the solution conflicting with $I$
\If{$I_{s} = \emptyset$ or ($I_{s}=\{I'\}$ and $I \subset I'$)}
    \If{ $I' \in M$}
        \State $M \gets M \cup \{I\}$
    \EndIf
    \State $S \gets S \cup \{I\} \setminus \{I'\} $\Comment{Take $I$ and discard $I'$ if necessary}
\ElsIf{$I$ is only involved in partial conflicts \textbf{and} $Prd(I) = 1$ \textbf{and} $I_{s} \cap M = \emptyset$}
    \State $S \gets S \cup \{I\} \setminus I_s $\Comment{Take $I$ and discard conflicting intervals}
    \State $M \gets M \cup \{I\}$
\EndIf
\end{algorithmic}
\end{algorithm}

\begin{theorem}
    Algorithm \ref{alg:unw-revoke} achieves $ALG \geq OPT - \eta$.
\end{theorem}
\begin{proof}
    We follow the same approach as in the proof of Theorem \ref{theo:unw-naive-pos}, mapping optimal intervals to intervals taken by the algorithm, and to error. The main difference is that because of revoking, this mapping might be redefined throughout the execution of the algorithm. As before, we let $H$ be the set of error elements, and $H_{I} \subseteq H$ be the set of error elements introduced by $\eta(I)$. Let $I_{opt}$ be an optimal interval. We will define an injective mapping $F: OPT \rightarrow ALG \cup H$ as follows: If $I_{opt}$ is taken by the algorithm, it is initially mapped onto itself. If $I_{opt}$ is later replaced, it must be because of a partial conflict (w.l.o.g. no interval is subsumed by an optimal interval) with a new interval $I'$ with $Prd(I') = 1$. In this case, $I_{opt}$ will be mapped to an error element in $H_{I'}$, or if no further optimal intervals that conflict with $I$ are yet to arrive, it will be mapped to $I$. In both, subcases it will never be remapped.\\
    Consider now the case of $I_{opt}$ being rejected upon arrival. This can only happen if it is involved in (at most two) partial conflicts. There are two possible cases. The first case is that $Prd(I_{opt}) = 0$, and therefore $|H_{I_{opt}}| = 1$, in which case we map $I_{opt}$ to the error element of its own prediction. The second case is that $Prd(I_{opt}) = 1$, but at least one of the conflicting intervals was marked. Let $I_{c}$ be one of the marked, partially conflicting intervals. If $I_{c}$ was marked by being taken through a partial-conflict replacement, it means that $Prd(I_{c}) = 1$, and $|H_{I_{c}} \cup \{I_c\}| > 0$, in which case we can map $I_{opt}$ to an element $h_{c} \in H_{I_{c}}\cup \{I_c\}$.\\
    If $I_{c}$ was instead marked by a proper-inclusion-replacement, we trace the original interval that got marked through a partial-conflict-replacement. Call that interval $I^{'}_{c}$. It holds that $I^{'}_{c}$ conflicts with $I_{c}$, and therefore also conflicts with $I_{opt}$. Moreover, for $I^{'}_{c}$ to have been accepted, it must be that $Prd(I^{'}_{c})=1$ and $|H_{I^{'}_{c}} \cup \{I^{'}_{c}\}|>0$. In this case, we map $I_{opt}$ to an element $h_{c'} \in H_{I^{'}_{c}} \cup \{I^{'}_{c}\}$. In conclusion, we have that $ALG + |H| \geq OPT$, and we get the desired bound.
\end{proof}

We note that the performance of algorithm \ref{alg:unw-revoke} on the instance of Theorem \ref{theo:unw-naive-neg} is exactly equal to $OPT - \eta$, and we get the following lemma.

\begin{lemma}
    The performance of algorithm \ref{alg:unw-revoke} cannot be better than $OPT-\eta$.
\end{lemma}

We next show that the robustness of algorithm \texttt{Revoke-Unit} nearly matches the optimal online guarantee.

\begin{theorem} \label{theo:unw-rev-robust}
    With at most $k$ distinct interval lengths, algorithm \ref{alg:unw-revoke} is $(2k+1)$-robust.
\end{theorem}
\begin{proof}
    We use a charging argument and show that an interval taken by the algorithm can be charged by at most $2k+1$ optimal intervals. As soon as an optimal interval arrives, we map it to an interval already taken by the algorithm or itself. When an interval is replaced during the execution, all optimal intervals charged to it up to that point, will now be charged to the new interval that was accepted. We build upon the proof of Theorem 3.2 in \cite{borodin2023any}. In the case of the \textit{BK2K} algorithm (\cite{borodin2023any}), it is true that for every predecessor trace $\mathcal{P}$, and consecutive intervals $(I_i,I_{i+1}) \in \mathcal{P}$, $\Phi(I_{i+1}) \leq \Phi(I_i) + 2$, and the length of every predecessor trace is at most $k$. While the former is still true for algorithm \texttt{Revoke-Unit}, the latter is not, and that is because we have an additional replacement rule. However, we argue that for every $I_i \in \mathcal{P}$, if $I_j,\; j > i$ is the next interval in the trace that was accepted through proper-inclusion, it is true that $TC(I_j) \leq TC(I_i) + 3$.\\
    
    Figure \ref{fig:pcr-charge} shows how to maximize charge on the event of a partial-conflict replacement. Before being replaced by a partially-conflicting interval, $I_{1}$ can be directly charged by at most two optimal intervals ($I^{1}_{opt}, I^{2}_{opt}$), one on each side. After $I_{2}$ replaces $I_{1}$, it can also be directly charged by two optimal intervals, but only if $I^{2}_{opt}$ was not charged to $I_{1}$ earlier. In other words, if $I_{2}$ is directly charged by two new intervals, it means that $I_{1}$ was directly charged by at most one, concluding that $\Phi(I_2) \leq TC(I_1) + 3$.

    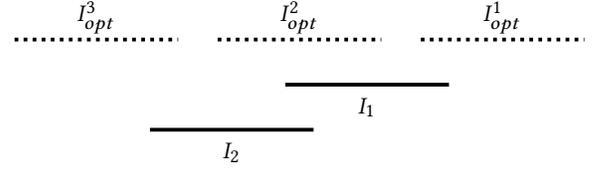
\begin{figure}[H]
	\centering
	
	\begin{tikzpicture}[scale=0.6]

	\node at (-5,-1.5) {$I_{2}$};
	\node[draw=none] (I1a) at (-7,-1) {$ $};
	\node[draw=none] (I1b) at (-3,-1) {$ $};
	\draw[line width=0.5mm] (I1a) -- (I1b);


	\node[draw=none] (I2a) at (-5.5,1) {$ $};
	\node[draw=none] (I2b) at (-1.5,1) {$ $};
	\node at (-3.5,1.5) {$I^{2}_{opt}$};
	\draw[dotted,line width=0.5mm] (I2a) -- (I2b);
	
	\node[draw=none] (I3a) at (-10,1) {$ $};
	\node[draw=none] (I3b) at (-6,1) {$ $};
	\node at (-8,1.5) {$I^{3}_{opt}$};
	\draw[dotted,line width=0.5mm] (I3a) -- (I3b);
	
	\node[draw=none] (I4a) at (-4,0) {$ $};
	\node[draw=none] (I4b) at (0,0) {$ $};
	\node at (-2,-0.5) {$I_{1}$};
	\draw[line width=0.5mm] (I4a) -- (I4b);
	
	\node[draw=none] (I5a) at (-1,1) {$ $};
	\node[draw=none] (I5b) at (3,1) {$ $};
	\node at (1,1.5) {$I^{1}_{opt}$};
	\draw[dotted,line width=0.5mm] (I5a) -- (I5b);

	\end{tikzpicture} 
	\caption{Maximum charge through partial-conflict replacement.}\label{fig:pcr-charge}
\end{figure}

Finally, notice that because the mark of an interval carries over when it is replaced, the event of a partial-conflict replacement can occur at most once in each predecessor trace,
and excluding at most one subsequence $(I_i, I_r, I_j) \in \mathcal{P}$ where $TC(I_j) \leq TC(I_i) + 3$, it holds that for $(I_b,I_{b+1}) \in \mathcal{P}$, $\Phi(I_{b+1})\leq \Phi(I_b) + 2$, giving us a worst case competitive ratio of $2k+1$. 
\end{proof}

\begin{corollary}
    With at most $k$ distinct interval lengths, and predictions with total error $\eta$, Algorithm Revoke-Unit achieves $ALG \geq \max\{OPT-\eta , \frac{OPT}{2k+1}\}$.
\end{corollary}

Notice how we can choose not to carry over the mark when proper-inclusion replacement occurs, and get a $3k$-robust algorithm. Such an algorithm is prone to follow the prediction more often, and it can outperform \texttt{Revoke-Unit} for some small values of error caused by adversarial predictions.\\\\
We now look at the case of proportional weights. In the conventional online setting, Garay et al. \cite{garay1997efficient} give a $2\phi + 1  \approx 4.236-$competitive algorithm, while Tomkins \cite{tomkins1995lower} gives a matching lower bound. They call their optimal algorithm \texttt{LR} (for \textit{length of route}), and we include it here for completeness. Unlike the case of unit weights, we now want to accept intervals that occupy as much of the line as possible. Algorithm \texttt{LR} works greedily by always accepting a new interval with no conflicts, and when there are conflicts, it accepts the new interval if its length is at least $\phi$ times greater than the largest conflicting interval. More generally, using parameter $\beta \geq \phi$, we have the following lemma:
\begin{lemma}[Garay et al. \cite{garay1997efficient}]
    Algorithm \texttt{LR} with parameter $\beta \geq \phi$ is $(2\beta + 1)$-competitive for the problem of interval selection with proportional weights.
\end{lemma}
\begin{algorithm}
\caption{LR \cite{garay1997efficient}}\label{alg:garay_prop}
\begin{algorithmic}
\State Parameter $\beta = \phi$ \Comment{optimal value for parameter $\beta$}
\State On the arrival of $I$:
\State $I_{s} \gets $ Set of intervals currently in the solution conflicting with $I$
\If{ $w(I) > \beta \cdot \max\{w(J)\; : \; J\in I_s\}$}
    \State Accept $I$ and displace conflicts
    \State Return
\EndIf
\end{algorithmic}
\end{algorithm}
Instead of using algorithm \texttt{LR} as the base of our predictions algorithm, we consider a slightly modified version, which we refer to as \texttt{LR$'$}, and which compares the weight of the new interval to the sum of the weights of the conflicting intervals, instead of looking only at the longest interval. Although we do not know the exact performance of algorithm \texttt{LR$'$} in the online model, we conjecture it is also $(2\phi + 1)$-competitive.\\\\
In trying to utilize predictions in the case of proportional weights, we first make the following observations:
\begin{observation}
    $1$-consistency is unattainable while maintaining bounded robustness.
\end{observation}
\begin{proof}
    To be $1$-consistent, the algorithm must be able to replace an interval with an arbitrarily smaller one that is part of the optimal solution. The adversary could then stop the instance, forcing arbitrarily bad robustness.
\end{proof}

\begin{observation}
    In order to have bounded robustness, it must be that a new interval that is sufficiently large (small\footnote{More accurately, an interval reducing ALG sufficiently much. }) must always be accepted (rejected).
\end{observation}

\begin{definition}[$\alpha-$increasing]
    An $\alpha$-increasing algorithm never accepts a new conflicting interval
that is less than $\alpha$ times the longest interval it conflicts with.
\end{definition}
\begin{lemma}
    An $\alpha$-increasing algorithm (greedy or non-greedy), cannot be better than $(2\alpha +1)$-consistent.
\end{lemma}
\begin{proof}
    Let an interval $I_1$ arrive first. Let $I_2$ and $I_3$ be intervals that partially conflict with $I_1$ on either side, and $w(I_2) = w(I_3) = \alpha \cdot w(I_1) - \epsilon$. Let $I_4$ with $w(I_4) = w(I_1)-2\epsilon$ be an interval that is fully subsumed by $I_4$. This instance is depicted in figure \ref{fig:neg-alpha-increasing}. The algorithm will never replace $I_1$, while the optimal solution is made of $\{I_2,I_3,I_4\}$.
    \begin{figure}[h] 
        \centering
        \begin{tikzpicture}[scale=0.45]
	
	\node at (0,0.5) {$I_1$};
	\node[draw=none] (I1a) at (-3,0) {$ $};
	\node[draw=none] (I1b) at (3,0) {$ $};
	\draw[line width=0.5mm] (I1a) -- (I1b);

 \node at (-6.3,-0.5) {$I_2$};
	\node[draw=none] (I2a) at (-9.3,-1) {$ $};
	\node[draw=none] (I2b) at (-2.3,-1) {$ $};
	\draw[line width=0.5mm] (I2a) -- (I2b);

  \node at (6.3,-0.5) {$I_3$};
	\node[draw=none] (I2a) at (2.3,-1) {$ $};
	\node[draw=none] (I2b) at (9.3,-1) {$ $};
	\draw[line width=0.5mm] (I2a) -- (I2b);

	\node[draw=none] (I5a) at (-2.5,-1) {$ $};
	\node[draw=none] (I5b) at (2.5,-1) {$ $};
	\node at (0,-1.5) {$I_{4}$};
	\draw[line width=0.5mm] (I5a) -- (I5b);

	\end{tikzpicture} 
        \caption{Consistency bound for $\alpha$-increasing algorithms.}
        \label{fig:neg-alpha-increasing}
    \end{figure}
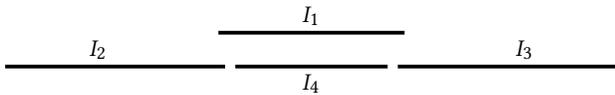
\end{proof}
Algorithm \texttt{LR} is a $\phi$-increasing algorithm, while the algorithm by Woeginger \cite{woeginger1994line} for the real-time model is $2$-increasing. Our predictions algorithm \ref{alg:prop-revoke2} \texttt{Revoke-Proportional} is $1$-increasing. The algorithm works like \texttt{LR$'$}, with one additional replacement rule that accepts a new interval that is predicted to be optimal, even if it is not sufficiently larger than what it conflicts with. More precisely, if a new interval is predicted to be optimal and is at least as big as the sum of the weights of the intervals it conflicts with, and none of the conflicting intervals were predicted to be optimal, it will be accepted through the predictions rule. The algorithm takes a parameter $\lambda > 1$, which can be thought of as an indicator of how much the predictions are trusted. As $\lambda$ increases, the consistency bound improves.

\begin{algorithm}
\caption{\texttt{Revoke-Proportional} {\hfil Parameter: $\lambda > 1$ }}\label{alg:prop-revoke2}
\begin{algorithmic}
\State On the arrival of $I$:
\State $I_{s} \gets $ Set of intervals currently in the solution conflicting with $I$
\State Let $w_c = \sum_{J \in I_s} w(J)$ \Comment{Total weight of conflicting intervals}
\If{ $w(I) \geq \lambda\cdot w_c$} \Comment{Main replacement rule}
    \State Accept $I$ and displace conflicts
    \State Return
\ElsIf{$Prd(I) = 1$} \Comment{Predictions rule} 
    \If{($w(I) \geq w_c$ and $|\{J: J \in I_s \text{ and }Prd(J)=1\}| = \emptyset$)} 
    \State Accept $I$ and displace conflicts
    \State Return
    \EndIf
\EndIf
\end{algorithmic}
\end{algorithm}

\begin{theorem}
    Algorithm Revoke-Proportional is $\frac{3\lambda}{\lambda -1}$-consistent.\label{theo:prop-consistent}
\end{theorem}
\begin{proof}
We consider the optimal solution $OPT$ consistent with the fully accurate predictions. We will show that throughout the execution of the algorithm, we have that $\Phi(I) \leq \mu \cdot w(I)$, for every $I$ in the current solution. In the end, we have that $\sum_{I\in ALG} \Phi(I) = OPT$, giving us the $\mu$-consistency of the algorithm.\\\\
As in the proof of Theorem \ref{theo:unw-rev-robust}, we consider the notions of \textit{transfer charge} ($TC$), and \textit{direct charge} ($DC$). We can express $\Phi(I) = TC(I) + DC(I)$. A transfer charge occurs whenever accepting a new interval $I$ replaces intervals currently in the solution. In that case, the total charge of those conflicting intervals is passed on as transfer charge to $I$. Any additional charge to $I$ after its acceptance is through direct charge, namely rejection of subsequent optimal intervals conflicting with $I$. We will write $DC_J(I)$ to denote the amount of direct charge from interval $J$ to interval $I$. Whenever an optimal interval is accepted, we consider its weight being directly charged to itself, and it cannot be directly charged again.\\\\
Whenever an optimal interval is rejected upon arrival, we charge its weight to the intervals it conflicts with, with its weight being distributed to all its conflicting intervals, in proportion to their weight. Specifically, let $I_o$ be the newly arrived optimal interval that is rejected, and $I_s$ denote the set of conflicting intervals. Each interval $J \in I_s$ is directly charged $DC_{I_o}(J)=w(I_o) \frac{w(J)}{w_c}\leq w(I_o)$. Furthermore, for an optimal interval to have been rejected, it must be that even the predictions rule failed, and because the predictions are accurate, it must have failed because $w(I_o) < w_c$. Because of this, we get that $w(I_o) \frac{w(J)}{w_c} \leq w(J)$, and therefore $DC_{I_o}(J)\leq\min\{w(I_o),w(J)\}$. An interval $I \in ALG$ can be directly charged by at most three different types of optimal intervals: 1) smaller intervals that are subsumed by it, 2) an optimal interval partially conflicting on the left, and 3) an optimal interval partially conflicting on the right.
In the case of smaller optimal intervals subsumed by $I$, the total amount of direct charge from those intervals can be at most $w(I)$. Given that each of the two possible partially conflicting intervals can directly charge $I$ at most $w(I)$, we conclude that for every $I\in ALG$:
\begin{equation}
\label{dc-bound}
    DC(I) \leq 3w(I)
\end{equation}
We omitted the case where the rejected optimal interval subsumes $I$, because in that case $DC(I) = w(I)$ and \ref{dc-bound} holds trivially. We now focus on the total amount of charge on any interval $I\in ALG$. Let:
$$\mu = \frac{3\lambda}{\lambda - 1}$$
We want to make sure that throughout the execution of the algorithm, $\Phi(I) \leq \mu \cdot w(I)$. Before any interval is accepted through replacement, intervals in the solution could have only been directly charged through rejected optimal intervals, and because of \eqref{dc-bound}, and the fact that $\lambda > 1$, our desired bound holds. We now consider all the cases of an interval being accepted through replacement.\\\\
\underline{Case $1$}: $I$ is an optimal interval and it is accepted through the predictions rule. In this case we have that $DC(I) = w(I)$, and we need to look at $TC(I)$. Let $L_c$ and $R_c$ denote the intervals (if any) that $I$ is partially conflicting with on the left and on the right respectively, and let $M_c$ denote the set of intervals that $I$ subsumes. We know that all of these conflicting intervals are not optimal, and they were accepted through the algorithm's main rule. First, notice that for all $J\in M_c$, $DC(J) = 0$, and $\Phi(J) = TC(J) \leq \frac{\mu}{\lambda} \cdot w(J)$. Moreover, $L_c$ and $R_c$ had not yet been directly charged by a partially conflicting optimal interval on one side, and therefore we have that $\Phi(L_c) \leq \frac{\mu}{\lambda}\cdot w(L_c) + 2w(L_c)$, and similarly $\Phi(R_c) \leq \frac{\mu}{\lambda}\cdot w(R_c) + 2w(R_c)$.\\
Putting everything together:
\[
\begin{aligned}
    TC(I) &= \sum_{J\in M_c} \Phi(J) + \Phi(L_c) + \Phi(R_c) \\
     & \leq \frac{\mu}{\lambda}\cdot w_c + 2(w(L_c) + w(R_c))\\
     & \leq \left(\frac{\mu}{\lambda} + 2\right)w_c\\
     &\leq \left(\frac{\mu}{\lambda} + 2\right)w(I)
\end{aligned}
\]
The last inequality being true from the fact that the main predictions rule is satisfied. Given also that $DC(I) = w(I)$, we get that $\Phi(I) \leq (\frac{\mu}{\lambda} + 2)w(I) + w(I) = (\frac{\mu}{\lambda} + 3)w(I)$. With our choice of $\mu$, we have:
\[\begin{aligned}
    \Phi(I) &\leq \left(\frac{\frac{3\lambda}{\lambda - 1}}{\lambda} + 3 \right)w(I) \\
    &= \left(\frac{3\lambda}{\lambda -1}\right)w(I)\\\\
\end{aligned}\]
\underline{Case $2$}: $I$ is an optimal interval and it is accepted through the algorithm's main rule. This is similar to case 1, with $DC(I) = w(I)$ and $w(I)\geq \lambda \cdot w_c$. The same analysis gives us $TC(I)\leq \left( \frac{\mu}{\lambda} + 2 \right)\frac{w(I)}{\lambda}$, and because $\lambda > 1$, the same bound holds.\\\\
\underline{Case $3$}: $I$ is not an optimal interval and it is accepted through the algorithm's main rule. In this case we have that $DC(I) \leq 3w(I)$, and we get that
\[\begin{aligned}
    \Phi(I) & \leq \sum_{J\in I_s} \Phi(J) + 3w(I)\\
    & \leq \frac{\mu}{\lambda}\cdot w(I) + 3w(I)\\
    & = \left(\frac{3\lambda}{\lambda -1}\right)w(I)
\end{aligned}\]
In conclusion, we have that throughout the execution of the algorithm, for $I \in ALG$, $\Phi(I)\leq \frac{3\lambda}{\lambda - 1}w(I)$, and therefore $\frac{OPT}{ALG} \leq \frac{3\lambda}{\lambda - 1}$. 
\end{proof}
\vspace{0.5cm}
We see that as $\lambda \rightarrow \infty$, the algorithm's consistency goes to $3$. We now look at the robustness of algorithm \texttt{Revoke-Proportional}.\\
\begin{theorem}
    Algorithm Revoke-Proportional is $\frac{4\lambda^2 + 2\lambda}{\lambda -1}$-robust.
\end{theorem}
\begin{proof}
    The argument is similar to the proof of Theorem \ref{theo:prop-consistent}. Both \textit{direct}, and \textit{transfer} charging work the same way as before. Let $\mu = \frac{2\lambda^2 +3\lambda + 1}{\lambda -1 }$, and $\delta = 2\lambda + 1$. We will show that that for every $I \in ALG$, $\Phi(I)\leq (\mu + \delta)\cdot w(I)=\frac{4\lambda^2 + 2\lambda}{\lambda -1}w(I)$.\\\\
    Notice first that the upper bound on direct charging is not as good as before. More precisely, with $I_o$ being a newly arrived optimal interval that will be rejected and $I_s$ being its conflicting intervals currently in the solution, we have that for every $J\in I_s$, $DC_{I_o}(J)= w(I_o)\frac{w(J)}{w_c}\leq \lambda \cdot w(J)$. More generally, $DC_{I_o}(J) \leq \min\{w(I_o),\lambda\cdot w(J)\}$. As before, given the three different possible types of conflicts, we have that:
    \begin{equation}
        DC(I) \leq (2\lambda + 1)w(I)
    \end{equation}
    We can now bound the total amount of charge on every interval in the algorithm's solution, throughout its execution. Before any replacement happens, the bound $\Phi(I) \leq (\mu + \delta)\cdot w(I)$ holds trivially.\\\\
    \underline{Case $1$}: Interval $I$ is accepted through the algorithm's main rule. We get that:
    \[\begin{aligned}
        \Phi(I) &\leq (\mu + \delta)\cdot w_c + (2\lambda + 1)\cdot w(I)\\
        & \leq (\mu + \delta)\cdot \frac{w(I)}{\lambda} + (2\lambda + 1)\cdot w(I)\\
        & = \left(\frac{\mu + \delta}{\lambda} + 2\lambda + 1\right)w(I)\\
        & = \left( \frac{\frac{4\lambda^2 +2\lambda}{\lambda-1}+2\lambda^2 + \lambda}{\lambda} \right)w(I)\\
        & = \mu \cdot w(I)
    \end{aligned}\]
\underline{Case $2$}: Interval $I$ is accepted through the algorithm's predictions rule. Notice that in this case, all conflicting intervals must have been accepted through the main rule, and not the predictions rule. Because of this, as we showed in case $1$, for every $J\in I_s$, it holds that $\Phi(J) \leq \mu\cdot w(J)$. This helps us bound the amount of transfer charge to interval $I$.
\[\begin{aligned}
    \Phi(I) &\leq \mu \cdot w_c + (2\lambda + 1)\cdot w(I)\\
    & \leq \mu \cdot w(I) + (2\lambda + 1)\cdot w(I) \\
    & = (\mu + \delta)\cdot w(I)
\end{aligned}\]
To summarize, we have shown that in the worst case, $\Phi(I) \leq (\mu + \delta)\cdot w(I)$ for every $I\in ALG$. This concludes the proof.
\end{proof}

\begin{figure*}[t!]
\centering
\caption{\texttt{NASA-iPSC} dataset.} (a) Unit \& Irrevocable, (b) Unit \& Revoking, (c) Proportional \& Irrevocable, (d) Proportional \& Revoking
\includegraphics[width=\textwidth]{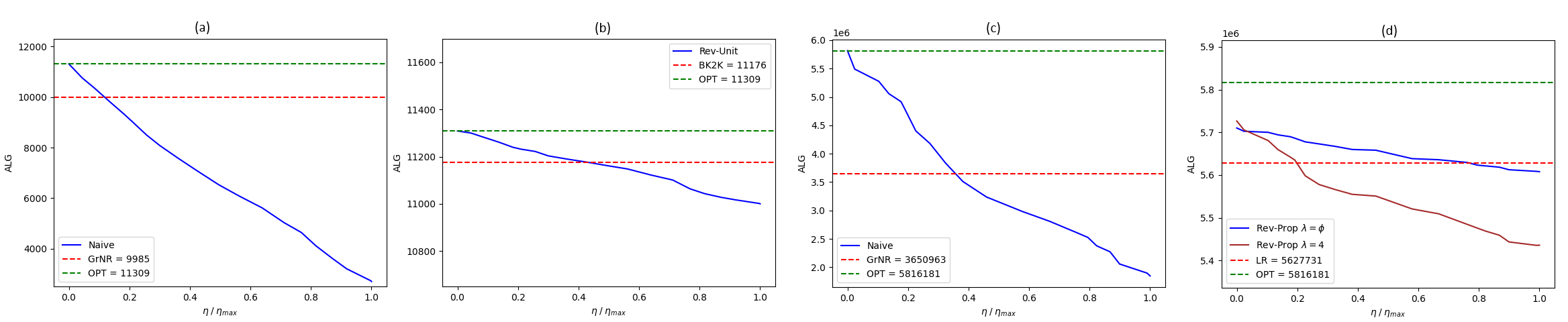}
\label{fig:nasa_exps}
\end{figure*}
\begin{figure*}[t!]
\centering
\caption{\texttt{CTC-SP2} dataset.}
\includegraphics[width=\textwidth]{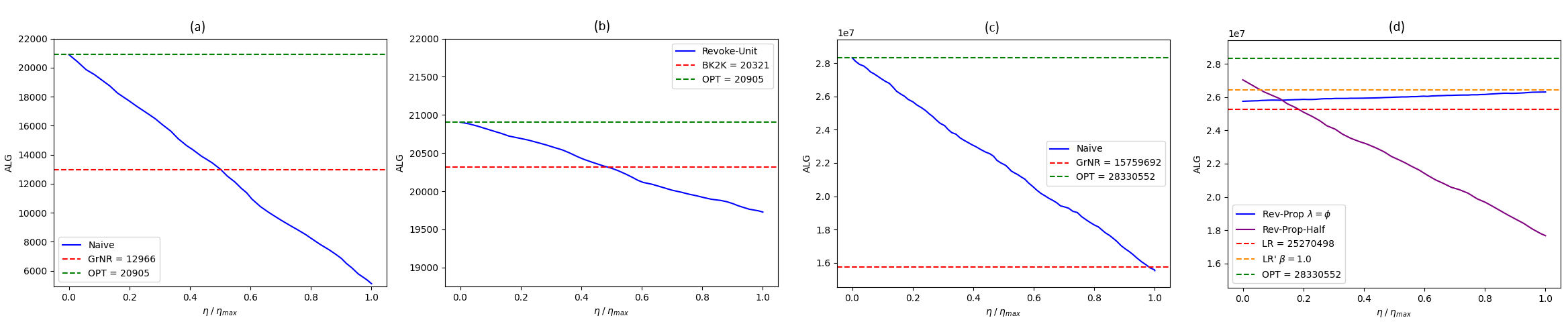}
\label{fig:ctc_exps}
\end{figure*}

We note that for $\lambda > \frac{2+\sqrt{5}}{\sqrt{5} -1}\approx 3.42 $, the consistency of our algorithm is already better than $2\phi + 1$, and $22.15$-robust. We have shown we can get consistency better than the online bound of \texttt{LR}, while maintaining bounded robustness. We believe further improvement on the bounds of \texttt{Revoke-Proportional} is possible, with an analysis that looks more closely at the dependence between direct and transfer charging.\\
One may also be able to further improve the algorithm by accepting an interval that is not as big as the sum of its conflicts, making the algorithm $a$-increasing with $a<1$. This would relax the predictions rule further, and make the algorithm more prone to bad choices caused by misleading predictions. In our experiments, we briefly discuss one such algorithm, which we call \texttt{Revoke-Prop-Half}, and which can accept a supposedly optimal interval even if it is half as big as its conflicts.


\section{Experimental Results}\label{section:exp}

We use real-world data from scheduling jobs on parallel machines\footnote{https://www.cs.huji.ac.il/labs/parallel/workload/} to test our algorithms. More information on the handling of these datasets can be found in a study by Feitelson et al. \cite{feitelson2014experience}. We focus on two datasets, \texttt{NASA-iPSC} (18,239 jobs) and \texttt{CTC-SP2} (77,222 jobs). As is usually the case, the performance of algorithms is much better than their worst-case bounds. For every algorithm we average its performance over random permutations of the input instance, for multiple error values. The $y$ axis values for proportional weights are expressed in scientific notation. We note that algorithm \texttt{GrNR} refers to a greedy algorithm without revoking, a very natural algorithm to compare our \texttt{Naive} algorithm against. All other algorithms have been mentioned earlier in the paper. Our experimental results are in line with our intuition, with the predictions algorithms outperforming predictionless algorithms for some values of the error, even when they are not $1$-consistent. Especially in the setting of revocable acceptances, even with half of the max possible error, our predictions algorithms perform just as well as their purely online counterparts. In the CTC dataset (figure \ref{fig:ctc_exps}) this is always the case, with the \texttt{Naive} algorithm outperforming \texttt{GrNR} for nearly all values of error.
In the case of proportional weights with revoking in figure \ref{fig:nasa_exps}, it is noteworthy that the variant of \texttt{Revoke-Proportional} with $\lambda = 4$, outperforms the $\lambda = \phi$ variant for some small values of error, but its performance degrades faster. This further validates the notion that the bigger the $\lambda$, the more the algorithm follows the predictions.\\
We also have to address the seemingly abnormal behavior of algorithm \texttt{Revoke-Proportional} in figure \ref{fig:ctc_exps}(d). As the error increases, so does the performance of \texttt{Revoke-Proportional}, which is counterintuitive and dissimilar to the corresponding plot of figure \ref{fig:nasa_exps}. This is because of the underlying structure of the \texttt{CTC-SP2} dataset, on which \textit{greedy} algorithms perform exceptionally well. We showcase this by having included algorithm \texttt{LR$'$} with $\beta = 1$, the algorithm that accepts a new interval if it is at least as big as everything it conflicts with. As the error increases, a larger number of intervals can be accepted through this clearly beneficial, relaxed predictions rule, which helps explain the improved performance. We also contrast this with algorithm \texttt{Rev-Prop-Half}, which uses a modified predictions rule, that can accept supposedly optimal intervals that are half the weight of their conflicts. This makes the algorithm more sensitive to the predictions, and its performance falls in line with what we would expect.\\
In conclusion, algorithms for interval selection can greatly benefit from utilizing imperfect predictions, and remain robust even in the presence of high error.


\begin{acks}
The author would like to thank Allan Borodin, Joan Boyar, and Kim Larsen for many helpful discussions, and for pointing out errors in earlier versions of this work.
\end{acks}



\newpage

\bibliographystyle{ACM-Reference-Format} 
\bibliography{thesis}


\end{document}